\documentclass[reqno,12pt,letterpaper]{amsart}

\usepackage{amsmath,amssymb,amsthm,graphicx,url}
\usepackage[usenames,dvipsnames]{color}
\usepackage[colorlinks=true,linkcolor=Red,citecolor=Green]{hyperref}

\usepackage{microtype}

\def\?[#1]{\textbf{[#1]}\marginpar{\Large{\textbf{??}}}}

\let\epsilon=\varepsilon 
\newtheorem*{lemma*}{Lemma}
\newtheorem*{theorem*}{Theorem}
\newtheorem*{proposition}{Proposition}
\numberwithin{equation}{section}
\DeclareMathOperator{\spec}{Spec}
\DeclareMathOperator{\tr}{tr}
\DeclareMathOperator{\res}{Res}
\newcommand{\inp}[1]{\left\langle #1 \right\rangle}
\newcommand{\w}{\omega}

\title[Trace formulas for TBG]{Trace formulas for TBG}
\author{Henry Zeng}
\address{Massachusetts Institute of Technology, Cambridge, MA 02139, USA}
\date{}
\thanks{This note is based on work done as an undergraduate at the University of California, Berkeley.}

\begin{document}

\begin{abstract}
Becker et al \cite{Becker_2022} produced a striking trace formula 
for the sum of fourth powers of 
generalized magic angles, $ \theta $, for the chiral model \cite{Tarnopolsky_2019} of twisted bilayer
graphene (TBG) with the exact Bistritzer--MacDonald potential \cite{Bistritzer_2011} :  $\sum \theta^{4} = 8 \pi / \sqrt 3 $. The purpose of this note is 
to generalize this formula to a larger class of potentials satisfying the symmetries of the model.
\end{abstract}

\maketitle

\section{Introduction and statement of result}
We consider the chiral model of twisted bilayer graphene
\begin{equation}
    H(\alpha) = \begin{pmatrix}
        0 & D(\alpha)^* \\ D(\alpha) & 0
    \end{pmatrix},\quad
    D(\alpha) = \begin{pmatrix}
        2D_{\bar z} & \alpha U(z)\\ \alpha U(-z) & 2D_{\bar z}
    \end{pmatrix},\quad \alpha \in \mathbb C,
    \label{eq:BMH}
\end{equation}
where $z = x+iy \in \mathbb C$, $D_{\bar z} = \frac{1}{2i}(\partial_x + i\partial_y)$, 
and $U$ is a complex potential.
The parameter $\alpha$ corresponds to the inverse angle of twisting.
Our goal is to generalize a formula for the sums of inverse powers of magic $\alpha$, 
previously obtained when $U$ is the Bistritzer-MacDonald potential.

Let
\begin{equation}
    \Lambda = \mathbb Z \oplus \omega\mathbb Z \subset \mathbb C,\quad\, \omega = e^{2\pi i/3}.
\end{equation}
Fix the inner product $\inp{z,w} = \text{Re}\,(z\overline w)$.
Then the dual lattice is
\begin{equation}
    \Lambda^* = \{\lambda \in \mathbb C: \text{ for all }z \in \Lambda,\, \inp{z,\lambda} \in 2\pi \mathbb Z\} = -\frac{4\pi i}{\sqrt 3}\Lambda.
\end{equation}
Let $K = \frac{4\pi}{3}$. 
We assume the potential $U$ satisfies
\begin{align}
    U(z+\lambda) &= e^{i\inp{\lambda, K}}U(z) \quad
    \text{ for all }\lambda \in  \Lambda^*,\label{eq:latticesym}
    \\
    U(\omega z) &= \omega U(z).
    \label{eq:rotsym}
\end{align}

For $k \in \mathbb C/\Lambda^*$, let
$L^2_k$ to be the space of locally square-integrable functions on $\mathbb C$ such that
\(
    u(z+\lambda) = e^{i\inp{k,\lambda}}u(z),
\)
equipped with the natural inner product.
Let
\begin{equation}
    L^2_k(\mathbb C; \mathbb C^2)
    = L^2_{k-K}\oplus L^2_{k+K}.
\end{equation}
Similarly, define Sobolev spaces $H^1_k$ and $H^1_k(\mathbb C;\mathbb C^2)$.
By the symmetry \eqref{eq:latticesym},
$D(\alpha)$ defines an operator from $H_k^1(\mathbb C;\mathbb C^2)$ to $L_k^2(\mathbb C;\mathbb C^2)$.

Using the second symmetry \eqref{eq:rotsym}, 
the following property was shown in \cite{Becker_2022}:
\begin{align}
\label{eq:flatbands}
    &k = \pm K \,\,(\text{mod }\Lambda^*)\quad \Rightarrow \quad \ker_{H^1_k(\mathbb C;\mathbb C^2)} D(\alpha) \ne 0\,\text{ for all }\alpha \in \mathbb C,\\
    &k \ne \pm K \,\,(\text{mod }\Lambda^*)\quad \Rightarrow \quad 
    \left( \ker_{H^1_k(\mathbb C;\mathbb C^2)}D(\alpha) \ne 0
    \text{ if and only if } \alpha \in \mathcal A\right).\notag
\end{align}
Here $\mathcal A \subset \mathbb C$ is a set not depending on $k$.

The set $\mathcal A$ determined by the above property is known as the set of \textit{magic} $\alpha$.
There is a notion of multiplicity for
magic $\alpha$ --- see section \ref{s:trace_intro}.
The set $\mathcal A$ 
for the Bistritzer-MacDonald potential
is shown in Figure \ref{f:angles}.

\begin{figure*}
\centering
    \includegraphics[scale=0.7]{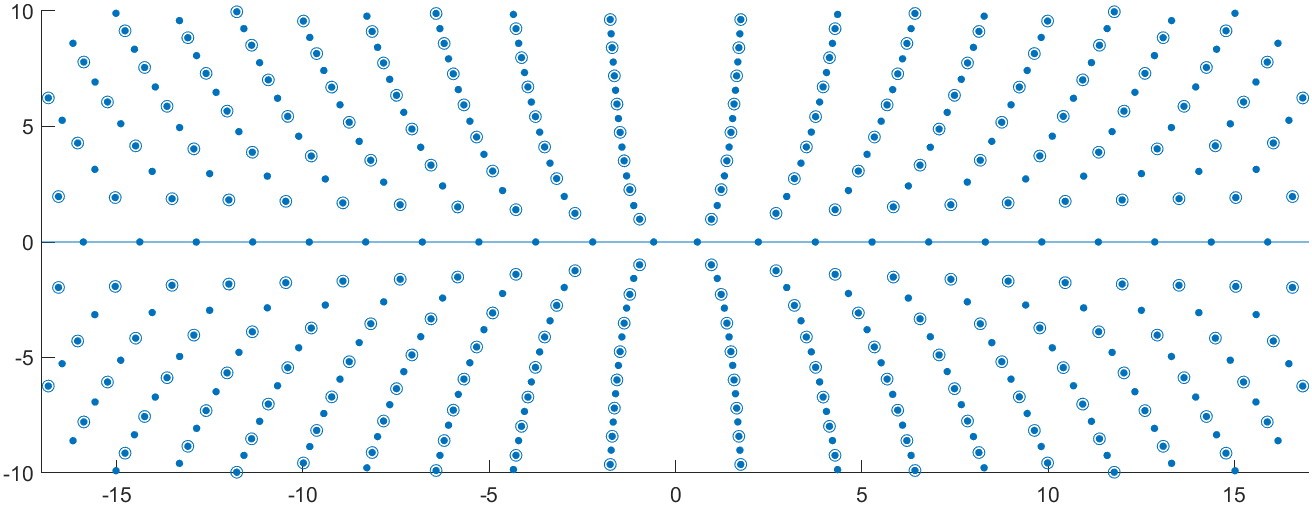}
    \caption{The set of magic $\alpha$,
    depicted with multiplicity.
    }
    \label{f:angles}
\end{figure*}


In general a potential $U$ satisfying \eqref{eq:latticesym}
and \eqref{eq:rotsym}
has a decomposition in the Fourier basis
  \begin{equation}
U(z) = \sum_{l \in S} c_l U_l(z), \quad\quad U_l(z) = e^{i\inp{z, l}} + \w e^{i\inp{z, \w l}} + \w^2 e^{i\inp{z, \w^2 l}},  \label{eq:U_family_0}
\end{equation}
where $S$ is a set of representatives for the  equivalence classes
of $K+\Lambda^*$ under rotation by $2\pi/3$.
We restrict to the following subclass of potentials:
\begin{equation}
  U(z) = \sum_{n \in 1 + 3 \mathbb Z} c_nU_{nK},\quad\quad
  c_n =0\textrm{ for all but finitely many $n$}.
 \label{eq:U_family}
\end{equation}

\begin{theorem*}
  For a potential of the form \eqref{eq:U_family},
  \begin{equation}
  \label{eq:t1}
  \sum_{\alpha \in \mathcal A}\alpha^{-4} = 
   \frac{8\sqrt 3 \pi}{K^4}\sum_{n,m} \frac{c_n^2c_m^2 + 2 c_n^2c_mc_{-n-m}}{n^2 + nm + m^2},
  \end{equation}
  where the sum on the left is taken with multiplicities
  in the sense of section \ref{s:trace_intro}.
\end{theorem*}

\section{Spectral characterization and trace formulas}

\label{s:trace_intro}
In this section we review some previous results.
Note that multiplication by $e^{i\inp{k,z}}$
maps $L^2_0(\mathbb C;\mathbb C^2)$
to $L^2_k(\mathbb C;\mathbb C^2)$,
and we have
\begin{equation}
    e^{-i\inp{k,z}}D(\alpha)(e^{i \inp{k,z}}u) = (D(\alpha) + k)u.
\end{equation}
Now for $k \not\in \{-K,K\} + \Lambda^*$,
$2D_{\bar z} + k$ is invertible on $L^2_0(\mathbb C; \mathbb C^2)$ and we can factorize
\begin{equation}
  D(\alpha) + k = (2D_{\bar z}+k) (1 + \alpha T_k),\quad\quad
  T_k := (2D_{\bar z}+k)^{-1} 
  \begin{pmatrix}
    0 & U(z)\\ U(-z) & 0
  \end{pmatrix}.
\end{equation}
The magic angles are then the inverse eigenvalues of a compact operator $T_k$:
\begin{equation}
\alpha \in \mathcal A \Leftrightarrow -\frac{1}{\alpha} \in \spec_{L_0^2(\mathbb C;\mathbb C^2)} T_k.   \label{eq:compact_spec}
\end{equation}
This was first observed in  \cite{Becker_2022}.
By the property \eqref{eq:flatbands}, we could take in \eqref{eq:compact_spec}
any $$k \not\in \{-K,K\} + \Lambda^* .$$
In the present note, by multiplicity of a magic angle we will always mean
algebraic multiplicity as an eigenvalue of $T_k$.
This is independent of the choice of $k$.

In view of \eqref{eq:compact_spec},
we can try to study the trace of powers of $T_k$.
In particular, for $s \ge 3$, the operator $T_k^{s}$
is trace-class (\cite{Becker_2022}).
Note that for odd values of $s$,
the operator $T_k^s$ has only off-diagonal components,
and has trace zero.
Thus one is really interested in the trace of powers of
\begin{equation}
\begin{gathered} 
T_k^2 = 
\begin{pmatrix}
T_k^+  & 0\\
  0 & T_k^-
\end{pmatrix},\quad\quad
T_k^\pm :=
  (2D_{\bar z}+k)^{-1}U(\pm z) (2D_{\bar z}+k)^{-1}U( \mp z).
  \end{gathered}
\end{equation}
Since the trace in invariant under cyclic permutations,
define for 
$k \not\in \{K,-K\} + \Lambda^*$
the Hilbert-Schmidt operator
\begin{equation}
  A_k := (2D_{\bar z} + k)^{-1} U(z) (2D_{\bar z} + k)^{-1} U(-z): L^2_{-K} \rightarrow L^2_{-K}.
\end{equation}
Let $p \ge 2$. Then by Lidskii's theorem
\begin{equation}
\label{eq:deftau}
    \tr A_k^{p} = \frac{1}{2}\tr T_k^{2p} = \frac{1}{2}\sum_{\alpha \in \mathcal A} \alpha^{-2p},
\end{equation}
where the sum is taken with algebraic multiplicity. Denote this sum \eqref{eq:deftau} by $\tau_p$.

  
Calculating the trace, the following result was shown.
\begin{proposition}[\cite{Becker_2022}, \cite{becker2023integrability}]
  For the Bistritzer-MacDonald potential
  \begin{equation}
    \sum_{\alpha \in \mathcal A}\alpha^{-4} = \frac{8 \pi}{\sqrt 3},
  \end{equation}
  and moreover for $p \in \mathbb N + 2$,
  \begin{equation}
    \sum_{\alpha \in \mathcal A}\alpha^{-2p} \in \frac{\pi}{\sqrt 3}\mathbb Q.
  \end{equation}
  
\end{proposition}

Various additional results
regarding the trace for general potentials 
satisfying \eqref{eq:latticesym} and \eqref{eq:rotsym}
were obtained in \cite{becker2023integrability}.
They include a formula for the trace as a sum of residues (see section 3.2),
and a result similar to the second part of the above proposition, for general potentials.
We recommend interested readers to take a look.

We now prove the theorem. 
We will compute the trace of $A_k^2$ using the Fourier basis.
The calculation will be simplified by taking advantage of the invariance of the trace for different values of $k$.

\section{Proof of Theorem}
\subsection{$A_k^p$ in fourier basis}

For section 3.1, fix
$k \in \mathbb C \setminus (\{K,-K\} + \Lambda^*)$
and $p \ge 2$.

Consider the orthonormal basis $\{e_l\}_{l \in -K+\Lambda^*}$ of $L^2_{-K}$:
\begin{equation}
e_l := \frac{\sqrt 2}{3^{1/4}}\,e^{i\inp{z,l}}, \quad\quad l \in -K+\Lambda^*.
\end{equation}
We have
\begin{equation}
\text{tr}(A_k^p) = \sum_{l \in -K+\Lambda^*} \inp{A_k^p e_l, e_l}.
\end{equation}

Introduce operators $D$ and $J_a$,
\begin{align}
  &D: L^2_s \rightarrow L^2_s, \quad\quad  De^{i\inp{z,l}}:= (l+k)^{-1}e^{i\inp{z,l}}  \,\,\,\text{ for }l \in s+\Lambda^*,\\
  &J_a: L^2_s \rightarrow L^2_{s+a}, \quad\quad J_au := e^{i\inp{a,z}}u.
\end{align}
So
\begin{equation}
(2D_{\bar z} + k)^{-1} = D,\quad U_l(z) = J_{l} + \w J_{\w l} + \w^2 J_{\w^2 l}.
\end{equation}
Letting
\begin{equation}
  D_a:L^2_s \rightarrow L^2_s, \quad\quad  D_a e^{i\inp{z,l}}:= (l+k-a)^{-1}e^{i\inp{z,l}}  \,\,\,\text{ for }l \in s+\Lambda^*,\\
  \label{eq:shifted_D}
\end{equation}
note that
\begin{equation}
  J_a D_b = D_{a+b} J_a,   \label{eq:comm}
\end{equation}
which along with $J_aJ_b = J_{a+b}$ and $D_aD_b = D_bD_a$ gives relations for the algebra
of operators generated by $D$ and the $J_a$  (with product given by composition).

Suppose
$U(z)$ is a potential of the form \eqref{eq:U_family_0}
with finitely many nonzero Fourier modes.
Then
\begin{equation}
    A_k = 
    \sum_{\alpha \in S}
    \sum_{\beta \in S}
    \sum_{\gamma \in \{0,1,2\}}
    \sum_{\delta \in \{0,1,2\}}
    c_{\alpha} c_{\beta}\,\w^{\gamma + \delta}
  D J_{\alpha \w^{\gamma}} D J_{-\beta \w^{\delta}}.
\end{equation}
Here $S$ is a set of representatives for equivalence classes of $K + \Lambda^*$ under multiplication by
$\omega$.
So
\begin{equation}
  A_k^p = \sum_{\alpha \in S^p}
    \sum_{\beta \in S^p}
    \sum_{\gamma \in \{0,1,2\}^p}
    \sum_{\delta \in \{0,1,2\}^p} \prod_{j=1}^p\left(  c_{\alpha_j} c_{\beta_j} \,\w^{\gamma_j + \delta_j}
  D J_{\alpha_j \w^{\gamma_j}} D J_{-\beta_j \w^{\delta_j}} \right),
\end{equation}
where $\alpha_j$
refers to the $j$th component
of $\alpha \in S^p$, and similarly for
$\beta_j, \gamma_j, \delta_j$.

Applying \eqref{eq:comm} and grouping terms,
\begin{align}
  A_k^p 
  &= \sum_{\alpha, \beta \in S^p}
    \sum_{\gamma,\delta \in \{0,1,2\}^p}
  \left(\prod_{j=1}^p c_{\alpha_j} c_{\beta_j}\right) \w^{\sum_{j=1}^p (\gamma_j + \delta_j)}
  \left(\prod_{j=1}^p D_{\tilde \alpha_j} D_{\tilde \beta_j}\right) J_{\sum_{j=1}^p (\alpha_j \w^{\gamma_j} - \beta_j \w^{\delta_j})} 
\end{align}
where
\begin{equation}
  \tilde \alpha_j = \sum_{i=1}^{j-1} (\alpha_i \w^{\gamma_i} - \beta_i \w^{\delta_i}),\quad\quad
  \tilde \beta_j = \alpha_j\w^{\gamma_j} +  \sum_{i=1}^{j-1} (\alpha_i \w^{\gamma_i} - \beta_i \w^{\delta_i}).
\end{equation}

Define
\begin{equation}
\Theta_p = \left\{(\alpha,\beta,\gamma,\delta) \in S^p \times S^p \times \{0,1,2\}^p\times \{0,1,2\}^p\,:\, \sum_{j=1}^p (\alpha_j \w^{\gamma_j} - \beta_j \w^{\delta_j}) = 0\right\}.
\end{equation}

Then by \eqref{eq:shifted_D},
\begin{align}
\inp{A_k^p e_l,e_l} &= \sum_{\Theta_p}
  \left(\prod_{j=1}^p c_{\alpha_j} c_{\beta_j}\right) \w^{\sum_{j=1}^p (\gamma_j + \delta_j)}
  \inp{  \left(\prod_{j=1}^p D_{\tilde \alpha_j} D_{\tilde \beta_j}\right) e_l, e_l}\\
  &= \sum_{ \Theta_p}
  \left(\prod_{j=1}^p c_{\alpha_j} c_{\beta_j}\right) \w^{\sum_{j=1}^p (\gamma_j + \delta_j)}
  \prod_{j=1}^p \frac{1}{(k+l-\tilde \alpha_j)(k+l-\tilde \beta_j)}.   \label{eq:fourierAkp}
\end{align}

\subsection{A useful formula for the trace}

Using the fact that the trace of $A_k^p$ is the same for different values of $k\not\in \{-K,K\} + \Lambda^*$,
a formula for the trace
was obtained in Theorem 4 of \cite{becker2023integrability}.
It is given (with slight modification) below.
\begin{lemma*}
  Suppose $U(z)$ is a potential 
  satisfying \eqref{eq:latticesym} and \eqref{eq:rotsym},
  with finitely many nonzero Fourier modes.
  Let $p \ge 2$.
  The function
  $f: \mathbb C \setminus (\{-K,K\} + \Lambda^*) \rightarrow \mathbb C$,
  \begin{equation}
    f(k) := \inp{A_k^p e_{-K}, e_{-K}},
  \end{equation}
  is meromorphic on $\mathbb C$
  with a finite number of poles,
  and we have
  \begin{equation}
    \tau_p =
    -\frac{\sqrt{3}}{8\pi} \sum_{b \in \{-K,K\}+\Lambda^*} \overline b \,\, \res(f,b).
  \end{equation}
  (Note the sum is finite.)
\end{lemma*}

We reproduce the proof here for clarity (as there is a slight modification).

\begin{proof}

  The fact that $f$ is meromorphic
  with a finite number of poles
  follows
  from the observation that the sum in \eqref{eq:fourierAkp}
  is finite and each term is meromorphic
  with a finite number of poles.
  Note also from \eqref{eq:fourierAkp} that
  $f(k) = O(k^{-2p})$.
  
  For any $k \in \mathbb C \setminus (\{-K,K\} + \Lambda^*)$, we know
  \begin{equation}
    \tau_p = \sum_{l \in -K + \Lambda^*} \inp{A_k^p e_l, e_l}
    = \sum_{\gamma \in \Lambda^*} f(k+\gamma),
  \end{equation}
  where the sum converges absolutely.
  This implies
  \begin{align}
    3 \sqrt 3 K \tau_p
    = \sum_{l=0}^2  \int_{[0,i\w^l \sqrt{3}K]}\left(\sum_{\gamma\in\Lambda^*} f(k+\gamma)\right) ds.
    \label{eq:threeint}
  \end{align}
  Here and for the rest of this proof, we will use the notation
  $\int_{[a,b]} g(k) ds$ to mean
  the integral over the arc-length parameter $s$
  parameterizing $k$ on
  the straight line segment from $a$ to $b$.

  On the other hand, by representing the residues as contour integrals over triangles,
  \begin{align}
    2\pi i \sum_{b \in \{-K,K\} + \Lambda^*} \overline b\,\, \res(f, b)
    &= \sum_{b \in K + \Lambda^*} \overline b \,\, \sum_{l=0}^2 \int_{[-\w^lK, -\w^{l+1}K]} f(k+b) \cdot i\w^{l-1}ds\\
    &\quad + \sum_{b \in -K + \Lambda^*} \overline b \,\, \sum_{l=0}^2 \int_{[\w^lK, \w^{l+1}K]} f(k+b) \cdot -i\w^{l-1} ds \notag\\
    & \hspace{-2cm} = \sum_{r \in \{1,-1\}} \sum_{\gamma \in \Lambda^*} \sum_{l=0}^2 ri\w^{l-1}\overline{(\gamma+rK)}
    \int_{[\gamma + r(K -\w^lK), \gamma + r(K-\w^{l+1}K)]} f(k) ds. \label{eq:bigsum}
  \end{align}
  For the second equality 
  we use the fact that the sum in \eqref{eq:bigsum} converges absolutely,
  which follows from $f(k) = O(k^{-2p})$.
  Grouping pairs of terms with the same integral,
  the sum becomes
  \begin{align}
    &= \sum_{\gamma \in \Lambda^*} \sum_{l=0}^2 i\w^{l-1} (-\w^{-l+1}K)
    \int_{[0, i\w^{l+1}\sqrt{3}K]} f(k+\gamma) ds\\
    &= -iK \sum_{\gamma \in \Lambda^*}\sum_{l=0}^2 \int_{[0, i\w^l \sqrt 3 K]}f(k+\gamma)ds.
  \end{align}
  Comparing this and \eqref{eq:threeint}, we see
  \begin{equation}
    \tau_p = -\frac{2\pi}{3\sqrt 3 K^2} \sum_{b \in \{-K,K\}+\Lambda^*} \overline b \,\, \res(f,b),
  \end{equation}
  which is the desired result.
\end{proof}

Combining the lemma and equation \eqref{eq:fourierAkp}
gives
\begin{equation}
  \tau_p =-\frac{\sqrt 3}{8 \pi}
  \sum_{\Theta_p}
  \left(\prod_{j=1}^p c_{\alpha_j} c_{\beta_j}\right) \w^{\sum_{j=1}^p (\gamma_j + \delta_j)}
  \sum_{b \in \{0,K\}+\Lambda^*} \overline b \cdot \res \left(\prod_{j=1}^p
  \frac{1}{(z-\tilde \alpha_j)(z-\tilde \beta_j)}, b \right).
  \label{eq:theformula}
\end{equation}

\subsection{Calculation of trace}
We show the theorem by evaluating equation \eqref{eq:theformula} for $p=2$
and a potential of the form
\eqref{eq:U_family}:

  By doing casework, the sum over $\Theta_2$ in \eqref{eq:theformula}
  splits into
  sums for each of the following cases:
  
  $(\w^{\gamma_1} \alpha_1, \w^{\delta_1} \beta_1, \w^{\gamma_2} \alpha_2, \w^{\delta_2} \beta_2) =$
  \begin{enumerate}
  \item $(\w^{r} n_1 K,\,\, \w^{r} n_1 K,\,\, \w^{r+s} n_2 K,\,\, \w^{r+s} n_2 K)$ for some $r \in \{ 0,1,2\}, s \in \{1,2\}$, and
    $n_1,n_2 \in 1+ 3\mathbb Z $
  \item $(\w^{r} n_1 K,\,\, \w^{r} n_2 K,\,\, \w^{r} n_3K, \,\,\w^{r} n_4K)$ for some $r \in \{ 0,1,2\}$ and $n_1,\dots, n_4 \in 1+3\mathbb Z$
    satisfying $n_1 - n_2 + n_3 - n_4 = 0$.
  \item $(\w^{r} n_1 K,\,\, \w^{r+s} n_2 K,\,\, \w^{r+s} n_2 K,\,\, \w^{r} n_1 K)$ 
  for some $r \in \{ 0,1,2\}, s \in \{1,2\}$, and
    $n_1,n_2\in 1+3\mathbb Z$
  \item $(\w^{r} n_1 K,\,\, \w^{r+s} n_2 K,\,\, \w^{r+2s} n_1 K,\,\, -\w^{r+s} (n_1+n_2)K)$ 
  for some $r \in \{ 0,1,2\}, s \in \{1,2\}$, and
    $n_1,n_2\in 1+3\mathbb Z$
  \item $(\w^{r} n_1 K,\,\, \w^{r+s} n_2 K,\,\, -\w^{r} (n_1+n_2)K,\,\, \w^{r+2s} n_2K)$  
  for some $r \in \{ 0,1,2\}, s \in \{1,2\}$, and
    $n_1,n_2\in 1+3\mathbb Z$
  \end{enumerate}
  Indeed, either $\gamma_1 = \delta_1$, or not.
  If $\gamma_1 = \delta_1$,
  either $\gamma_2 \ne \gamma_1$ (case 1),
  or $\gamma_2 = \gamma_1$ (case 2).
  If $\gamma_1 \ne \delta_1$, then
  given $n_1, n_2 \in 1 + 3 \mathbb Z$,
  there are three possibilities
  corresponding
  to each choice of $\gamma_2$ (cases 3-5).

  We claim that in each case, different values of $r$ result in the same term in \eqref{eq:theformula},
  so we may calculate the sum for terms with $r = 0$ and multiply by $3$.
  Indeed, this follows from
 \[  \begin{gathered}
    \w^{\sum (\gamma_j+1) + (\delta_j+1)} = \w \cdot \w^{\sum \gamma_j+\delta_j},   \\
        \res\left(\prod_{j=1}^2
    \frac{1}{(z- \w\tilde\alpha_j)(z-\w \tilde \beta_j)}, b \right)
    = \res\left(\prod_{j=1}^2
    \frac{1}{(z-\tilde \alpha_j)(z-\tilde \beta_j)}, \w^{-1} b \right).
  \end{gathered} \]

  Also, for any $(\alpha,\beta,\gamma,\delta)$, note the sum of the residues of
 $$ z \prod \frac{1}{(z-\tilde \alpha_j) (z - \tilde \beta_j)}$$ is $0$.
  Thus, to slightly shorten the
  calculation we first subtract, for each $(\alpha, \beta, \gamma, \delta)$, the sum of the
  residues of $ z \prod \frac{1}{(z-\tilde \alpha_j) (z - \tilde
    \beta_j)} .$
  Then the residues on $\mathbb R \setminus \{0\}$ cancel, and we have
    \begin{align*}
\tau_2
&=
-\frac{3\sqrt 3}{8 \pi} \cdot \sum_{n,m} 
\Bigg(
c_n^2 c_m^2 \sum_{s =1,2 }
\w^{-s}\left(\frac{1}{nK \cdot \w^s mK} +\frac{(\w^{-s}- \w^s) m K}{(\w^s m K)^2 (\w^s m K - n K)} 
 \right)\\
&\quad\quad +c_n^2c_m^2\frac{1}{nK \cdot mK}
+c_n^2 c_m^2 \sum_{s =1,2 }
\frac{\w^{-s}(-\w^{-s}+ \w^s) m K }{(\w^s m K)^2 (-\w^s m K - n K)}\\
&\quad\quad + c_n^2c_mc_{-n-m}
\sum_{s=1,2}\bigg(\frac{\w^s(-\w^{-s} + \w^s)m K}{(-\w^s mK)(nK-\w^smK)(-\w^{-s}nK)}\\
&\quad\quad\quad\quad
\quad\quad\quad\quad +\frac{\w^s(\w^{s} - \w^{-s})(-n-m) K}{(\w^{-s} nK)\w^s(-n-m)K(\w^{-s}nK - \w^{s}mK)}
\bigg)
\\
&\quad\quad +  c_nc_m^2c_{-n-m}
\sum_{s=1,2}\bigg(
\frac{(\w^{-s}- \w^s) m K}{-\w^s mK (nK - \w^s mK)(n+m)K}
\\
&\quad\quad\quad\quad\quad
\quad\quad\quad
+\frac{(\w^{s}-\w^{-s}) m K}{\w^{-s}mK(-n-m)K(\w^{-s}mK - nK)}
\bigg)\Bigg)
\\
&=
-\frac{2\pi}{\sqrt 3 K^4}\sum_{n,m} \Bigg(
-c_n^2 c_m^2\frac{3}{n^2 + nm + m^2}
-c_n^2 c_m^2\frac{3}{n^2 + nm + m^2}\\
&\quad\quad\quad\quad\quad
-c_n^2 c_m c_{-n-m}\frac{6}{n^2 + nm + m^2}-c_n c_m^2 c_{-n-m}\frac{6}{n^2 + nm + m^2}
\Bigg)\\
&=
\frac{4\sqrt 3 \pi}{K^4}\left( \sum_{n,m} \frac{1}{n^2 + nm + m^2}c_n^2c_m^2
    + \sum_{n,m} \frac{2}{n^2 + nm + m^2} c_n^2 c_m c_{-n-m}
    \right).
\end{align*}


\subsection*{Acknowledgements}
I would like to thank Professor Zworski
for introducing me to magic angles, and being a wonderful mentor. 
I would also like to thank Tristan Humbert for reading a previous version of this writeup and providing many helpful comments.
I also acknowledge support from the Simons Foundation
through the Targeted Grant Award No.\ 896630.

\bibliographystyle{alpha}
\bibliography{ref}

\end{document}